\documentclass[11pt,a4paper]{article}

\usepackage[utf8]{inputenc}
\usepackage[T1]{fontenc}
\usepackage{amsmath,amssymb,amsthm,mathrsfs}
\usepackage{geometry}
\usepackage{hyperref}
\usepackage{cleveref}
\usepackage{enumitem}
\usepackage{booktabs}
\usepackage{array}
\usepackage{graphicx}
\usepackage{xcolor}
\usepackage{titlesec}
\usepackage{setspace}
\usepackage{fancyhdr}
\usepackage{natbib}

\newtheorem{postulate}{Postulate}[section]

\newtheorem{proposition}{Proposition}

\newtheorem{law}{Law}
\newtheorem{definition}{Definition}
\newtheorem{principle}{Principle}
\newtheorem{conjecture}{Conjecture}

\newcommand{\Phyg}{\mathscr{P}}
\newcommand{\Uspace}{\mathscr{U}}
\newcommand{\Dspace}{\mathscr{D}}
\newcommand{\Sspace}{\mathscr{S}}
\newcommand{\Fib}{\mathscr{F}}
\newcommand{\OmMass}{\boldsymbol{\mu}}
\newcommand{\PhyForce}{\boldsymbol{\Phi}}
\newcommand{\LieD}{\mathcal{L}}

\titleformat{\section}{\Large\bfseries}{\thesection.}{0.5em}{}
\titleformat{\subsection}{\large\bfseries}{\thesubsection.}{0.5em}{}
\titleformat{\subsubsection}{\normalsize\bfseries\itshape}{\thesubsubsection.}{0.5em}{}

\hypersetup{
  colorlinks=true,
  linkcolor=blue!60!black,
  citecolor=blue!60!black,
  urlcolor=blue!60!black
}

\begin{document}

% ---- Title Page ----
\begin{titlepage}
\centering
\vspace*{3cm}
{\Huge\bfseries The Unified Field Theory\\[0.3cm] of Phygital Space\par}
\vspace{0.8cm}
{\Large A Sheaf-Theoretic Framework with Finsler Geometry,\\
Autopoietic Dynamics, and Non-Equilibrium Thermodynamics\par}
\vspace{2cm}
{\Huge Silvio Meira\par}
\vspace{0.3cm}
{\normalsize silvio@meira.com\par}
\vspace{0.3cm}
{\Large TDS.company \quad$\vert$\quad cesar.school\par}
\vfill
{\normalsize v0.2, April 2026\par}
\end{titlepage}

% ---- Abstract ----
\begin{abstract}
\noindent
This paper proposes a Unified Field Theory of Phygital Space, positing that contemporary reality is not a dichotomy of ``online'' and ``offline,'' but a unified ontological manifold of irreducible but coupled dimensions. 

We formalize Phygital Space as a \emph{sheaf} over a topological site composed of the Physical~($\Uspace$), Networked Digital~($\Dspace$), and Networked Social~($\Sspace$) dimensions, grounded in \emph{Informaticity}---the triune capacity to compute, communicate, and control---and instantiated through \emph{Platforms}.

We develop a framework incorporating \emph{Finsler geometry} to model asymmetric cross-dimensional interaction costs, define \emph{Ontological Mass}~($\OmMass$) as a positive semidefinite tensor encoding directional resistance to change, and introduce \emph{autopoietic dynamics} to account for the endogenous agency of persons, algorithms, and social formations. 

We propose a non-equilibrium thermodynamic model where economic value is negentropy generated by platforms acting as \emph{dissipative structures}. We introduce \emph{Temporal Shear} formalized through Lie derivatives. 

The framework is applied analytically to the Chinese e-commerce ecosystem (1999--2025), extended to a post-human ecology of \emph{Synthetic Agents}, and developed into normative implications for platform governance and human flourishing. 

We present this as a \emph{formal research program}---a structured set of postulates, derived propositions, and testable conjectures---rather than a completed axiomatic system.
\end{abstract}

\newpage
\tableofcontents
\newpage

% ============================================================
% SECTION 0: INTRODUCTION
% ============================================================
\section{Introduction: Beyond Dualism}

The history of Western metaphysics, from the Cartesian \emph{res cogitans} and \emph{res extensa} to Castells's dialectic of the ``space of places'' and the ``space of flows,'' has struggled with the relationship between the material and the informational \citep{descartes1641, castells1996}. We note at the outset that Castells himself recognized the dialectical coexistence of these spaces---places and flows interpenetrate rather than partition reality \citep[ch.~6]{castells1996}. Nevertheless, even Castells's nuanced framework lacks the formal machinery to \emph{quantify} the friction of a digital transaction, the \emph{mass} of a social media influencer, or the \emph{entropy} of an online community. We inhabit a \textbf{Phygital Space}: a generative matrix where materiality, information, and sociality are inextricably co-constituted. What we lack is a formal language for its physics.

\subsection{Related Work and Positioning}

Several traditions have anticipated aspects of our framework. \citet{kitchin2011} developed the thesis that software and physical space are ``transduced'' into a single ontology, introducing the concept of \emph{code/space} where code and spatiality are mutually constituted. \citet{couldry2017} formulated the thesis of ``deep mediatization,'' arguing that all social construction of reality is now mediated by digital infrastructures. \citet{hayles1999} anticipated our discussion of synthetic agents by analyzing the conditions under which posthuman subjectivity emerges from human-machine symbiosis. 

\citet{floridi2014} introduced the concept of the ``infosphere'' and, in later work \citep{floridi2019}, developed a logic of information that formalizes some of the ontological claims we make here. \citet{latour2005} distributed agency between human and non-human actants in a manner that our Section~5 formalizes for synthetic entities. \citet{luhmann1995} extended autopoiesis from biological to social systems, a move we adopt and discuss critically. Our contribution is not the \emph{insight} that the physical, digital, and social are co-constituted---this has been established by the traditions cited above. 

Our contribution is \textbf{twofold}. 

First, we posit the existence of a unified \emph{Phygital Space}---a structured manifold of physical, digital, and social dimensions---in which agents with compound, cross-dimensional behaviors interact, compete, and co-evolve. Prior work has described the co-constitution of the physical and the digital \citep{kitchin2011}, the deep mediatization of social reality \citep{couldry2017}, and the distribution of agency across human and non-human actants \citep{latour2005}. But none has proposed a single space in which all three dimensions are formally integrated, where entities possess measurable coordinates in each dimension, and where the coupling between dimensions is encoded in a geometric structure that constrains and shapes all interaction. 

Second, we attempt to provide \emph{formal structure} for this space: a geometric framework (fiber bundles, Finsler metrics) and a dynamical framework (autopoietic equations of motion, non-equilibrium thermodynamics) capable of generating testable propositions and guiding empirical research. 

We offer this not as a completed axiomatic system in the Hilbertian sense, but as a \emph{formal research program} in the sense of \citet{lakatos1978}: a hard core of foundational postulates surrounded by a protective belt of auxiliary hypotheses, derived propositions, and conjectures that can be refined, tested, and---crucially---falsified.

\subsection{The Three Structural Challenges}

A theory of Phygital Space must overcome three challenges: (a)~the Physical, Digital, and Social dimensions are \emph{coupled}, not independent; (b)~entities possess \emph{endogenous agency}---persons, algorithms, and social formations generate motion without requiring external forces; and (c)~cross-dimensional movement costs are \emph{asymmetric}---abstraction (physical$\to$digital) differs categorically from materialization (digital$\to$physical).

We address (a) with a sheaf-theoretic fiber bundle, (b) with autopoietic dynamics drawing on Maturana and Varela's concept as extended by \citet{luhmann1995} to social systems (a contested extension whose limitations we acknowledge; see \citealp{mingers1995}), and (c) with Finsler geometry.

\subsection{The Central Thesis}

Phygital Space is a \emph{generative matrix} whose dimensional structure actively shapes being, behavior, and becoming:

\begin{description}[style=unboxed,leftmargin=0.5cm]
  \item[Ontology is Sheaf-Theoretic Geometry.] Dimensions are coupled; local coherence does not guarantee global extendability.
  \item[Interaction is Finslerian Physics.] Costs are asymmetric and direction-dependent.
  \item[Agency is Autopoietic.] Entities possess intrinsic dynamics; external forces modulate but do not exclusively drive trajectories.
  \item[Value is Negentropy in a Dissipative System.] Platforms are far-from-equilibrium structures.
  \item[Time is a Lie-Derivative Field.] Temporal Shear~$(\sigma_\tau)$ threatens the stability of entities navigating multiple temporalities.
\end{description}

\subsection{Structure of the Argument}

Section~1 establishes the geometry. Section~2 introduces autopoietic dynamics. Section~3 develops non-equilibrium thermodynamics. Section~4 formalizes Temporal Shear. Section~5 extends the model to Synthetic Agents. Section~6 applies the theory analytically to the Chinese e-commerce ecosystem (1999--2025). Section~7 articulates normative implications. The Conclusion synthesizes findings and the future research agenda.

% ============================================================
% SECTION 1: GEOMETRY
% ============================================================
\newpage
\section{The Postulates --- The Sheaf-Theoretic Geometry of Phygital Space}

We propose four foundational postulates. We call them \emph{postulates} rather than \emph{axioms} to signal our intent honestly: they function as the foundational assumptions of a research program, not as the axioms of a completed formal system from which all propositions can be deduced by pure inference. Euclid's postulates served a similar role for millennia before Hilbert's formalization; we aspire to the former, not the latter, at this stage of the theory's development. Where propositions \emph{can} be derived from the postulates, we derive them. Where they cannot yet be derived, we state them as conjectures.

\subsection{Postulate I: The Principle of the Fibered Manifold}

\begin{postulate}[The Fibered Manifold]\label{post:fiber}
Phygital Space~$(\Phyg)$ is a fiber bundle over the Physical dimension~$(\Uspace)$, with fibers at each point encoding the available Digital~$(\Dspace)$ and Social~$(\Sspace)$ states. An entity~$E$ in Phygital Space is a \emph{section} of this bundle.
\end{postulate}

\noindent\textbf{Formal Definition.} Let $\Uspace$ be a topological space representing the Physical dimension. Over each point $u \in \Uspace$, define a fiber $\Fib(u) = \Dspace(u) \times \Sspace(u)$. The total space is:
\begin{equation}
  \Phyg = \bigsqcup_{u \in \Uspace} \{u\} \times \Fib(u),
\end{equation}
with projection $\pi\colon \Phyg \to \Uspace$. The coupling is encoded by a connection~$\nabla$ on the bundle, specifying how the fiber ``twists'' as one moves through the base space \citep{nakahara2003}.

This structure accommodates \textbf{dimensional holonomy}: when an entity traverses a closed loop in $\Uspace$ (travels abroad and returns), its digital and social coordinates may not return to their original values. Holonomy is a precise measure of irreducible coupling between dimensions.

\subsection{Postulate II: The Sheaf Condition}

\begin{postulate}[The Sheaf Condition]\label{post:sheaf}
The assignment of Phygital states to regions of the manifold satisfies the sheaf condition: local data compatible on overlaps glue uniquely into global data. Obstructions to gluing exist and are topological invariants of the manifold \citep{bredon1997}.
\end{postulate}

The sheaf condition explains why ``scaling'' a platform is not replicating code. When Uber works in San Francisco but not in Riyadh, the obstruction is a cohomological class encoding regulatory, cultural, and infrastructural incompatibilities \citep{spivak2014}.

\subsection{Postulate III: The Finsler Metric of Interaction Cost}

\begin{postulate}[The Finsler Quasimetric]\label{post:finsler}
Distance in Phygital Space is a Finslerian measure---an asymmetric, direction-dependent cost function on the tangent bundle.
\end{postulate}

\noindent\textbf{Formal Definition.} A Finsler function $F\colon T\Phyg \to \mathbb{R}_{\geq 0}$ satisfies: (i)~$F(x, \lambda v) = \lambda\, F(x, v)$ for $\lambda > 0$; (ii)~the Hessian $g_{ij}(x,v) = \tfrac{1}{2}\,\partial^2 F^2 / \partial v^i \partial v^j$ is positive definite. We do \emph{not} require $F(x, v) = F(x, -v)$. Distance along $\gamma$ is:
\begin{equation}
  \delta(E_1, E_2) = \int_0^1 F\bigl(\gamma(t),\, \dot{\gamma}(t)\bigr)\, dt,
\end{equation}
and generically $\delta(E_1, E_2) \neq \delta(E_2, E_1)$ \citep{bao2000}.

Why Finsler rather than a generic quasimetric or a Lorentzian structure? A generic quasimetric provides no geometric structure---no geodesics, no curvature, no tangent-space decomposition. Finsler geometry is the \emph{minimal} generalization of Riemannian geometry that accommodates asymmetry while retaining differentiable structure rich enough to define geodesics (optimal paths through the manifold), curvature (how the presence of massive entities warps the cost landscape), and parallel transport (how dimensional states evolve along trajectories). Lorentz geometry would impose a specific signature $(-,+,+,\ldots)$ without justification. The positive homogeneity condition (i) is a simplifying idealization: it asserts that doubling the ``velocity'' of dimensional movement exactly doubles the cost. Empirically, this may not hold precisely---there may be economies or diseconomies of scale in cross-dimensional translation. We adopt it as a first approximation, noting that deviations from homogeneity are second-order corrections that do not alter the qualitative predictions.

\subsubsection{The Three Asymmetric Frictions}

The asymmetry of the Finsler function decomposes into three directional costs, each corresponding to a distinct type of cross-dimensional translation. These are not merely different magnitudes of the same process; they are categorically different operations, governed by different logics and subject to different constraints.

\medskip\noindent\textbf{The Friction of Abstraction} $(\eta_{\mathrm{phy}\to\mathrm{dig}})$. Moving from the Physical to the Digital requires \emph{abstraction}---the reduction of matter to data. As \citet{shannon1948} demonstrated, this process is inherently lossy: the qualia of the physical object (texture, temperature, weight, smell, the way light plays across a surface) are stripped away to fit within digital bandwidth. A photograph of a sculpture is not the sculpture; a digital twin of a factory is not the factory. The informational loss is:
\begin{equation}
  \mathcal{L}_{\mathrm{abs}} = u_{\mathrm{input}} - d_{\mathrm{output}},
\end{equation}
where the cost of minimizing $\mathcal{L}_{\mathrm{abs}}$ grows with the fidelity required. Low-resolution abstraction is cheap (a product listing with a single photo); high-resolution abstraction is expensive (a photogrammetric 3D scan with haptic feedback data). The friction of abstraction explains the persistence of the ``uncanny valley'' \citep{mori1970} and the limitations of virtual reality: perfect abstraction would require capturing and encoding the full physical state of an object, which is thermodynamically prohibitive. This friction is relatively stable and monotonically related to the complexity of the source object.

\medskip\noindent\textbf{The Friction of Materialization} $(\eta_{\mathrm{dig}\to\mathrm{phy}})$. Moving from the Digital to the Physical requires \emph{materialization}---the ``bits-to-atoms'' transition that creates massive resistance \citep{gershenfeld2005}. Digital bits are malleable, copyable at near-zero marginal cost, and not subject to gravity. Atoms are constrained by thermodynamics, friction, scarcity, and entropy. This asymmetry is categorical: it is far easier to photograph a sculpture than to 3D-print one from its digital model; far easier to list a product online than to deliver it to a doorstep in a remote village. The ``last mile'' problem in logistics is the canonical manifestation of materialization friction---the cost of bridging the final gap between digital order and physical possession. Critically, $\eta_{\mathrm{dig}\to\mathrm{phy}} \gg \eta_{\mathrm{phy}\to\mathrm{dig}}$ in virtually all practical cases, and this inequality is the single most important structural feature of the Phygital manifold. It explains why e-commerce platforms that master digital aggregation (low $\eta_{\mathrm{phy}\to\mathrm{dig}}$) still struggle with delivery (high $\eta_{\mathrm{dig}\to\mathrm{phy}}$), and why logistics infrastructure is the decisive competitive asset in mature phygital markets.

\medskip\noindent\textbf{The Friction of Legitimization} $(\eta_{\mathrm{soc}})$. Movement across \emph{any} dimension requires social legitimization---the consent of the social field to recognize the transaction as valid. A digital currency has no value ($\mu \to 0$) without a social consensus to accept it. A physical object has no phygital value without social attribution of meaning. A digital service has no market without trust. This friction is governed by what \citet{luhmann1995} called the ``double contingency'' of social systems: each party's willingness to act depends on their expectation of the other's willingness, creating a recursive dependency that can only be broken by trust, reputation, or institutional guarantee. The adoption curve of innovation \citep{rogers2003} is essentially a measurement of overcoming $\eta_{\mathrm{soc}}$---the progressive social legitimization of a new technology or practice from early adopters through to laggards. Social friction is often the \emph{highest} and most \emph{volatile} of the three: a platform can be technically flawless and logistically efficient, yet fail completely if it cannot generate social trust. Conversely, a technically inferior platform with strong social legitimization (community endorsement, regulatory approval, cultural resonance) can dominate one with superior technology. This is why social friction is the wild card of the Phygital manifold---the dimension most resistant to engineering and most susceptible to sudden collapse.

\begin{proposition}[Asymmetry of Geodesics]\label{prop:asymmetry}
Let $\gamma^*_{AB}$ be the geodesic (cost-minimizing path) from $E_A$ to $E_B$ in $\Phyg$, and $\gamma^*_{BA}$ the geodesic from $E_B$ to $E_A$. Under Postulate~III, generically $\gamma^*_{AB} \neq \gamma^*_{BA}$ as curves, and $\delta(E_A, E_B) \neq \delta(E_B, E_A)$.
\end{proposition}

\begin{proof}
In a Finsler manifold with $F(x,v) \neq F(x,-v)$, the geodesic equation $\ddot{x}^i + 2G^i(x,\dot{x}) = 0$ depends on the spray coefficients $G^i$, which are themselves direction-dependent via $F$. Reversing the direction of traversal changes the spray, and therefore changes the geodesic. The length functional is evaluated on different curves with different integrands, so the distances generically differ. This is a standard result in Finsler geometry \citep[Thm.~6.2.1]{bao2000}.
\end{proof}

This proposition has non-trivial content: it tells us that the \emph{optimal strategy} for moving from physical to digital differs from the optimal strategy for moving from digital to physical. Abstraction and materialization are not reverses of each other---they follow different paths through the manifold.

\subsection{Postulate IV: The Ontological Mass Tensor}

\begin{postulate}[The Mass Tensor]\label{post:mass}
Every entity~$E$ possesses an Ontological Mass Tensor~$\OmMass(E)$, a symmetric \textbf{positive semidefinite} matrix:
\begin{equation}
  \OmMass(E) = 
  \begin{pmatrix}
    \mu_{pp} & \mu_{pd} & \mu_{ps} \\
    \mu_{dp} & \mu_{dd} & \mu_{ds} \\
    \mu_{sp} & \mu_{sd} & \mu_{ss}
  \end{pmatrix},
  \quad \OmMass \succeq 0.
  \label{eq:mass-tensor}
\end{equation}
The \emph{rank} of $\OmMass(E)$ classifies the entity's dimensional presence: $\mathrm{rank}(\OmMass) = 3$ for entities with non-degenerate presence in all three dimensions; $\mathrm{rank}(\OmMass) = 2$ for entities absent from one dimension; $\mathrm{rank}(\OmMass) = 1$ for entities confined to a single dimension.
\end{postulate}

\noindent\textbf{Note on positive semidefiniteness.} We require $\OmMass \succeq 0$ (all eigenvalues $\geq 0$), not $\OmMass \succ 0$ (all eigenvalues $> 0$). This is essential: entities that lack physical embodiment (synthetic agents) have $\mu_{pp} = 0$, producing a rank-deficient tensor. A positive-definiteness requirement would exclude such entities from the theory. The rank classification introduces a natural taxonomy: biological entities with bodies, devices, and social ties have rank~3; synthetic agents have rank~2; a purely physical object with no digital or social presence has rank~1.

\begin{proposition}[Dimensional Lock-In]\label{prop:lockin}
If the off-diagonal coupling terms $\mu_{pd}$, $\mu_{ps}$, or $\mu_{ds}$ are large relative to the diagonal terms, then a force applied in one dimension produces acceleration in all coupled dimensions, with the cross-dimensional response governed by $\OmMass^{-1}$ (or its pseudoinverse for rank-deficient tensors).
\end{proposition}

\begin{proof}
By the equation of motion (Law~II, Section~2), $\vec{a} = \OmMass^{-1}\cdot\PhyForce$. For a $3\times 3$ positive semidefinite matrix, the inverse (or Moore-Penrose pseudoinverse for rank-deficient cases) couples the force components: $a_i = \sum_j (\OmMass^{-1})_{ij} \Phi_j$. When $\mu_{pd}$ is large, the $(p,d)$ and $(d,p)$ entries of $\OmMass^{-1}$ are non-negligible, so a force $\Phi_d$ in the digital dimension produces acceleration $a_p$ in the physical dimension and vice versa. The magnitude of cross-dimensional response is monotonically related to the off-diagonal coupling strength.
\end{proof}

\subsubsection{Summary of the Geometric Foundation}

We have defined: the \emph{structure} (fiber bundle with sheaf conditions), the \emph{ruler} (Finsler quasimetric with justified asymmetry), and the \emph{object} (positive semidefinite mass tensor with rank-based classification). From these postulates, we have derived two propositions: the generic asymmetry of geodesics (Proposition~\ref{prop:asymmetry}) and the mechanism of dimensional lock-in (Proposition~\ref{prop:lockin}).

% ============================================================
% SECTION 2: DYNAMICS
% ============================================================
\newpage
\section{The Physics --- Autopoietic Dynamics, Conservation, and Friction}

We replace Newtonian dynamics with an autopoietic framework. The term ``autopoietic'' requires careful qualification. \citet{maturana1980} defined autopoiesis for biological systems: operationally closed, self-producing networks of components. \citet{luhmann1995} extended the concept to social systems (communication-producing networks of communications). 

This extension has been criticized \citep{mingers1995} as stretching the concept beyond its biological grounding. We adopt the Luhmannian extension with the following caveat: we use ``autopoietic dynamics'' to mean that the entity possesses an intrinsic vector field on its state space---it generates its own motion from internal processes. For biological entities, this is autopoiesis in the strict sense. For algorithms, it is self-optimization. For social formations, it is norm drift and endogenous governance. The mathematical formalism is the same in all three cases; the ontological justification differs by entity type.

\subsection{The Three-Term Decomposition of Motion}

We propose that the trajectory of any entity~$E$ in Phygital Space is governed by a three-term equation of motion:
\begin{equation}
  \frac{dE}{dt} = \underbrace{f_{\mathrm{int}}(E)}_{\text{intrinsic}} 
  + \underbrace{\sum_j g_{\mathrm{coup}}(E, E_j)}_{\text{coupling}} 
  + \underbrace{\PhyForce_{\mathrm{ext}}}_{\text{external}},
  \label{eq:three-term}
\end{equation}
where each term captures a fundamentally different source of change.

The first term, $f_{\mathrm{int}}(E)$, is the \textbf{intrinsic dynamics}---the endogenous vector field generated by the entity's own processes. For a person, this includes biological drives, cognitive processes, intentional action, curiosity, boredom, and creative impulse. For an algorithm, this includes its optimization dynamics, learning processes, and self-modification through gradient descent or reinforcement learning. For a social formation (a community, an institution, a norm), this includes endogenous norm drift, collective memory evolution, and self-governance dynamics \citep{maturana1980, kauffman1993}. The intrinsic term ensures that entities move even in the absence of external forces---a feature categorically absent from Newtonian mechanics, where an object at rest remains at rest unless pushed. In Phygital Space, nothing is ever truly at rest: persons age, algorithms optimize, communities evolve, and norms drift.

The second term, $g_{\mathrm{coup}}(E, E_j)$, is the \textbf{coupling dynamics}---the mutual influence between entity~$E$ and other entities~$E_j$ in its neighborhood. This term captures network effects, social contagion, competitive pressure, imitation, and symbiotic relationships. It is not reducible to an external force because it is \emph{relational}: the interaction changes both~$E$ and~$E_j$ simultaneously. In network science terms, this is the adjacency-weighted influence propagating through the graph \citep{barabasi2002, watts1998}. The coupling dynamics explain why viral phenomena spread without any single ``push''---each infected node becomes a source of coupling force for its neighbors, creating positive feedback loops that can accelerate the entire system exponentially.

The third term, $\PhyForce_{\mathrm{ext}}$, is the \textbf{external forcing}---the truly exogenous perturbations: natural disasters, regulatory shocks, technological discontinuities, pandemics, and events originating outside the entity's network neighborhood. This is the only term that corresponds to the Newtonian ``external force,'' and it is now properly situated as \emph{one component among three}, not the sole driver of change. The distinction matters: a theory built on external forcing alone cannot explain why social movements grow, why algorithms improve, or why trust erodes---all phenomena driven by intrinsic and coupling dynamics.

\subsection{Law I: The Law of Autopoietic Inertia}

\begin{law}[Autopoietic Inertia]
An entity~$E$ possesses an intrinsic dynamics $f_{\mathrm{int}}(E)$ that evolves its coordinates without external forces. Its trajectory without perturbation is its \textbf{natural orbit}. Massive entities (high~$\OmMass$) are difficult to deflect from this orbit.
\end{law}

``Inertia'' in Phygital Space is not the tendency to remain at rest---it is the tendency to \emph{persist in one's own dynamics}. A social movement does not need an external push to grow; it has its own internal logic of mobilization \citep{tarrow1998}. An algorithm does not need an external trigger to update its parameters; its learning rule is self-executing. A community does not need an external shock to fragment; internal tensions accumulate and eventually tear through the social fabric. The concept of inertia is thus generalized from ``resistance to change'' to ``persistence of intrinsic character.'' Massive entities are not immobile---they are \emph{hard to redirect}. A nation-state has its own trajectory; it takes enormous force to alter it. A startup has almost no mass; it can pivot overnight. This is the Phygital analogue of Newton's insight that mass resists acceleration, but freed from the false assumption that mass prevents all self-generated motion.

\subsection{Law II: The Law of Tensorial Acceleration}

\begin{law}[Tensorial Acceleration]
For rank-3 entities: $\vec{a} = \OmMass(E)^{-1} \cdot \PhyForce$. For rank-deficient entities (rank~1 or~2): $\vec{a} = \OmMass(E)^{+} \cdot \PhyForce$, where $\OmMass^{+}$ is the Moore-Penrose pseudoinverse. Forces in the null space of $\OmMass$ produce no acceleration---the entity is ``transparent'' to forces in its absent dimensions.
\end{law}

Because $\OmMass$ is a tensor, not a scalar, the same force applied in different dimensional directions produces different accelerations. A regulatory shock (social force) may produce large acceleration in a startup (low~$\mu_{ss}$) but negligible acceleration in a nation-state (high~$\mu_{ss}$). Moreover, due to the off-diagonal coupling terms, a force applied in one dimension \emph{leaks} into others: a digital disruption (force in $\Dspace$) produces physical consequences (supply chain disruption) and social consequences (labor displacement) governed by the coupling terms $\mu_{pd}$ and~$\mu_{ds}$. The use of the Moore-Penrose pseudoinverse for rank-deficient entities is not merely a mathematical convenience---it captures the physical reality that an entity absent from a dimension cannot be affected by forces in that dimension.

\begin{proposition}[Dimensional Transparency]\label{prop:transparency}
An entity with $\mathrm{rank}(\OmMass) < 3$ is unaffected by forces applied purely in its absent dimensions.
\end{proposition}

\begin{proof}
If $\OmMass$ has rank~2 (e.g., $\mu_{pp} = 0$), then the null space of $\OmMass$ includes vectors aligned with the physical dimension. The pseudoinverse $\OmMass^+$ maps the null space to zero: $\OmMass^+ \cdot \PhyForce = 0$ when $\PhyForce$ lies entirely in $\ker(\OmMass)$. Physically: a synthetic agent with no physical mass is unaffected by earthquakes, weather, or physical barriers. Conversely, a purely physical object with no digital or social presence ($\mathrm{rank} = 1$) is unaffected by algorithmic updates or shifts in social norms. This is a non-trivial prediction with empirical content: it implies, for instance, that a regulation targeting physical infrastructure will have zero direct effect on a purely digital entity, and can only affect it indirectly through the coupling terms of entities that \emph{do} span both dimensions.
\end{proof}

\subsection{Law III: The Law of Asymmetric Reaction}

\begin{law}[Asymmetric Reaction]
For every action in one dimension, there is a reaction across intersecting dimensions, but the reaction is generically unequal in magnitude and non-opposite in direction, due to Finsler asymmetry and tensorial anisotropy.
\end{law}

Newton's Third Law posits equality of action and reaction. In Phygital Space, the Finsler asymmetry and tensorial mass ensure that reactions are systematically \emph{unequal}. When a force is applied to digitize a process (move in $\Dspace$), it generates counter-pressures in $\Uspace$ (hardware requirements, energy consumption, data center construction) and $\Sspace$ (displacement of labor, loss of tacit knowledge, erosion of face-to-face relationships). But these counter-pressures are not equal to the original force---they are filtered through the mass tensor's coupling terms and the Finsler metric's directional asymmetry. This ``reactive torque'' is consistently underestimated in practice. It explains why many digital transformation initiatives solve digital problems while creating \emph{disproportionate} social friction: the reactive force in the social dimension is often larger than the original force in the digital dimension, because $\eta_{\mathrm{dig}\to\mathrm{soc}}$ (the cost of social adaptation to digital change) exceeds $\eta_{\mathrm{dig}}$ (the cost of the digital change itself).

\subsection{The Principle of Approximate Conservation}

\begin{principle}[Approximate Conservation]
In mature phygital systems, total capacity for action is approximately conserved over time scales short relative to structural evolution.
\end{principle}

This principle requires careful statement about what it claims and what it does not. It \emph{claims} that the total capacity for action---the sum of physical labor capacity, computational capacity, and human attention---does not jump discontinuously within a stable system. Attention displaced from one platform reappears in another; computation shut down in one data center is compensated by computation spun up elsewhere; physical labor displaced by automation re-emerges in new forms of work. The principle \emph{excludes} models in which phygital energy can be created ex nihilo (a platform generating unbounded attention from a fixed population) or destroyed without trace.

The principle does \emph{not} claim exact conservation. We have not identified a Lagrangian or continuous symmetry group from which exact conservation could be derived via Noether's theorem \citep{noether1918}. The absence of a variational formulation is a limitation of the current theory, not a feature. Over longer time scales, new capacities can be created (a new medium opens new attentional channels; a new chip architecture creates new computational capacity) and existing capacities can be destroyed (cognitive decline, infrastructure collapse, demographic contraction). The principle is analogous to the adiabatic approximation in physics: conservation holds when the rate of systemic change is slow compared to the rate of internal dynamics.

\subsection{Anticipatory Dynamics}

A distinctive feature of entities in Phygital Space---and one that has no analogue in classical mechanics---is that they are \emph{anticipatory systems} \citep{rosen1985}. They possess internal models of their environment and act on predictions of future states, not merely on present conditions. The intrinsic dynamics decompose as:
\begin{equation}
  f_{\mathrm{int}}(E) = f_{\mathrm{reactive}}(E,\, \mathrm{state}) + f_{\mathrm{anticipatory}}(E,\, \mathrm{model}(\text{future})).
\end{equation}

The reactive component responds to the current state---a thermostat adjusting temperature, a reflex withdrawing a hand from heat. The anticipatory component acts on \emph{modeled} futures: a navigation app calculates predicted arrival times and alters the user's present behavior; a predictive algorithm forecloses options before they are consciously evaluated; a central bank raises interest rates based on projected inflation. Present behavior is shaped by predicted futures, which then confirm or falsify themselves in a feedback loop:
\begin{equation}
  \mathrm{Action}_{\mathrm{present}} = f\!\left(\mathrm{Prediction}_{\mathrm{future}}\right).
\end{equation}

This anticipatory structure connects naturally to the fiber bundle geometry of Section~1: the entity's internal model is a \emph{local section} of the bundle---a particular assignment of expected digital and social states to physical locations. Anticipatory action is navigation guided by this section. When the model is accurate, the entity follows an efficient path; when the model is wrong, the entity discovers the discrepancy only after committing resources, creating a sunk-cost dynamic that is characteristic of anticipatory failures. 

The anticipatory component is what makes algorithmic prediction commercially powerful and socially dangerous in equal measure: it can compress the Finsler distance between intent and action to near zero, but in doing so it collapses the deliberative space that human autonomy requires.

% ============================================================
% SECTION 3: THERMODYNAMICS
% ============================================================
\newpage
\section{The Thermodynamics --- Non-Equilibrium Value Creation}

We model platforms as \emph{dissipative structures} in the sense of \citet{prigogine1984}---far-from-equilibrium systems that maintain order through continuous energy throughput.

\textbf{A note on entropy concepts.} This section employs three distinct notions of entropy that must not be conflated. \emph{Thermodynamic entropy} (Boltzmann, 1877; measured in J/K) characterizes the physical disorder of material systems. \emph{Informational entropy} (Shannon, 1948; measured in bits) quantifies uncertainty in a probability distribution over messages. \emph{Social entropy} is an \emph{analogical} concept introduced here to denote the disorder, unpredictability, and loss of coherence in social configurations---measured operationally by variance in trust, fragmentation of norms, and cognitive load. We do \emph{not} claim that these are the same quantity or that they obey the same laws. We claim a \emph{structural analogy}: in all three cases, the creation of local order in one subsystem requires the export of disorder to the environment. The analogy is heuristic, not literal.

\subsection{The First Law: Transduction of Phygital Energy}

Phygital Energy~$(\Omega)$ exists in three forms: Physical~$(\Omega_{\mathrm{phy}})$---labor, material, electricity; Digital~$(\Omega_{\mathrm{dig}})$---computation, storage, bandwidth; and Social~$(\Omega_{\mathrm{soc}})$---attention and trust, bounded by biological cognitive limits and Dunbar's relational constraints \citep{simon1971}.

\begin{conjecture}[Approximate Conservation of Phygital Potential]\label{conj:conservation}
$\Omega_{\mathrm{total}} = \Omega_{\mathrm{phy}} + \Omega_{\mathrm{dig}} + \Omega_{\mathrm{soc}} \approx \mathrm{const.}$ over time scales short relative to systemic structural change.
\end{conjecture}

\noindent\textbf{Supporting argument.} Within a fixed time interval shorter than the characteristic time of structural change (new platform entry, demographic shift), the system's topology is approximately fixed: the number of humans bounds $\Omega_{\mathrm{soc}}$, installed computational capacity bounds $\Omega_{\mathrm{dig}}$, and physical infrastructure bounds $\Omega_{\mathrm{phy}}$. Transduction between forms is possible (attention converts to data; data to profit; profit to infrastructure), but each transduction conserves the total at the cost of friction losses (Postulate~III). Over longer scales, structural parameters change and conservation fails. We label this a \emph{conjecture} rather than a theorem because we have not derived it from a variational principle. The conjecture has empirical content: it predicts that attention displaced from one platform reappears in another rather than vanishing, and that total system capacity does not jump discontinuously.

\subsection{The Second Law: Dissipative Structures and Entropy Export}

\begin{proposition}[Negentropic Value Requires Entropy Export]\label{prop:entropy}
If a platform creates local order (reduces informational uncertainty for users, reduces search costs, coordinates supply and demand), then the combined system (platform + environment) experiences a net increase in total disorder. The platform exports this disorder as informational noise (spam, deprecated data, thermal waste) and social disruption (cognitive overload, erosion of privacy, polarization).
\end{proposition}

\begin{proof}
Model the platform as an open system importing high-grade energy ($\Omega_{\mathrm{soc}}$, $\Omega_{\mathrm{phy}}$) and outputting local order (matched transactions, reduced search entropy). By the Second Law of thermodynamics applied to the combined closed system (platform + environment), total entropy cannot decrease. Since the platform's internal entropy decreases (it becomes more organized), the environment's entropy must increase by at least the same amount. The exported entropy takes informational form (data noise, system complexity, thermal waste of computation) and social form (disrupted social structures, cognitive load on users).

This argument is structural, not quantitative: we invoke the Second Law at the level of the logical structure (local order requires environmental disorder), not at the level of specific entropy units. A fully quantitative treatment would require defining a common entropy measure across the three domains, which remains an open problem (see Section~8, Conclusion).
\end{proof}

This Prigoginean framework explains why digital utopias never materialize: digital order inevitably accelerates social and physical disorder elsewhere. The ``heat death'' of a platform occurs when exported social entropy exceeds the negentropic value it provides---a testable prediction.

\subsection{Platform Efficiency and the Definition of Phygital Value}

Platform efficiency: $\eta = V_{\mathrm{out}}/\Omega_{\mathrm{in}}$. Maximum efficiency is bounded by the biological regeneration rate of social substrate.

\begin{definition}[Phygital Value]\label{def:value}
\begin{equation}
  V_p = \mathrm{Tr}(\OmMass) \cdot I \cdot R,
  \label{eq:phygital-value}
\end{equation}
where $\mathrm{Tr}(\OmMass) = \mu_{pp} + \mu_{dd} + \mu_{ss}$ is a scalar measure of total dimensional stability, $I$ is Shannon information content, and $R$ is social resonance (probability of action).
\end{definition}

\noindent\textbf{Note on the trace.} A reviewer might ask: why the trace, which discards the off-diagonal coupling information? The trace captures the \emph{total} resistance to change, which is the relevant quantity for value persistence. The off-diagonal terms govern \emph{how} change propagates, not \emph{whether} the entity persists. The multiplicative form $\mathrm{Tr}(\OmMass) \cdot I \cdot R$ is stipulative at this stage---we propose it as a working definition, not a derived result. Its justification is that it correctly ranks known cases: misinformation (high $I$, high $R$, low trace) has high viral but low phygital value; a scientific result (high $I$, high trace, low $R$) has low market but high phygital value.

% ============================================================
% SECTION 4: CHRONOLOGY
% ============================================================
\newpage
\section{The Chronology --- Temporal Shear and the Lie-Derivative Formalism}

Phygital Time decomposes into three distinct regimes, each governed by a different logic and operating at a different characteristic velocity.

\textbf{Physical Time} $(\tau_{\mathrm{phy}})$ is the thermodynamic arrow---irreversible, entropic, linear, and continuous \citep{prigogine1996}. This is the time of biology, decay, and material fatigue: cells divide, muscles tire, buildings weather, food spoils. It is governed by the monotonic increase of entropy. Physical Time is the time of the body: circadian rhythms, seasonal cycles, aging, the irreducible fact that a human being must sleep, eat, and eventually die. Its characteristic velocity is bounded by biological processes---neuronal firing rates on the order of milliseconds, metabolic cycles on the order of hours, developmental processes on the order of years. Physical Time cannot be reversed, compressed, or paused. It is the most fundamental temporal regime, and the one most frequently ignored by digital systems designers.

\textbf{Digital Time} $(\tau_{\mathrm{dig}})$ is the computational arrow---discrete, compressible, and approaching instantaneity \citep{castells1996}. This is the time of the processor clock, the packet switch, the database query, and real-time data streaming. It creates what Castells called ``timeless time''---a regime in which the interval between cause and effect approaches zero. Unlike Physical Time, Digital Time is \emph{reversible}: computational states can be rolled back, cached, replayed, forked into parallel timelines. It is \emph{compressible}: a simulation can run a year of climate data in minutes. And it is \emph{parallelizable}: a million transactions can occur simultaneously. In its extreme form---Machine Time $(\tau_{\mathrm{cpu}})$---it operates at nanosecond scales, billions of times faster than any biological process. The gap between $\tau_{\mathrm{cpu}}$ and $\tau_{\mathrm{phy}}$ is perhaps the most consequential structural feature of the contemporary condition.

\textbf{Social Time} $(\tau_{\mathrm{soc}})$ is the narrative arrow---cyclical, rhythmic, and intersubjective \citep{elias1992, stiegler2010, rosa2013}. This is the time of rituals (annual holidays, election cycles, fashion seasons), collective memory (how long a society ``remembers'' an event), and institutional change (how quickly laws, norms, and governance structures adapt). Social Time is characterized by its \emph{viscosity}: social norms change slowly compared to digital updates, because they require coordination among millions of agents with heterogeneous beliefs and interests. But Social Time is also characterized by \emph{punctuated equilibria}: long periods of apparent stability interrupted by sudden, discontinuous shifts---revolutions, moral panics, paradigm changes, viral movements that reshape the social landscape in days. Social Time is fundamentally relational: it exists only through intersubjective coordination, the synchronization of expectations, and the shared construction of temporal horizons. A norm is not ``fast'' or ``slow'' in absolute terms; it is fast or slow relative to the expectations of the community that sustains it.

\subsection{Temporal Shear as a Lie Derivative}

\textbf{A note on the manifold structure.} The Lie derivative is defined on smooth manifolds \citep{lee2013}. We acknowledge that the full specification of the smooth structure of Phygital Space---its atlas, dimension, and differentiability class---remains an open problem. We work under the assumption that, at the scales relevant to our analysis, the space admits a smooth approximation sufficient for the Lie derivative to be well-defined. This is analogous to the use of continuum mechanics for granular media: the underlying reality may be discrete, but the continuum approximation captures the relevant macroscopic behavior. We flag this as a modeling assumption, not a proven fact.

Let $\tau_{\mathrm{dig}}$ and $\tau_{\mathrm{soc}}$ be vector fields representing temporal flows. Temporal Shear is:
\begin{equation}
  \sigma_\tau = \left\|\LieD_{\tau_{\mathrm{dig}}}\, \tau_{\mathrm{soc}}\right\|,
  \label{eq:temporal-shear}
\end{equation}
where $\LieD$ denotes the Lie derivative. This measures how the social temporal flow \emph{deforms} when dragged along the digital temporal flow. The formalism is coordinate-independent, directional, and connected to the bundle's connection (Section~1). The full Temporal Shear tensor captures all pairwise interactions:
\begin{equation}
  \boldsymbol{\sigma}_\tau = \begin{pmatrix}
    0 & \LieD_{\tau_{\mathrm{phy}}}\tau_{\mathrm{dig}} & \LieD_{\tau_{\mathrm{phy}}}\tau_{\mathrm{soc}} \\
    \LieD_{\tau_{\mathrm{dig}}}\tau_{\mathrm{phy}} & 0 & \LieD_{\tau_{\mathrm{dig}}}\tau_{\mathrm{soc}} \\
    \LieD_{\tau_{\mathrm{soc}}}\tau_{\mathrm{phy}} & \LieD_{\tau_{\mathrm{soc}}}\tau_{\mathrm{dig}} & 0
  \end{pmatrix}.
\end{equation}

\subsection{The Cost of Temporal Coherence}

\begin{proposition}[Synchronization Cost]\label{prop:sync}
The maintenance of a unified Phygital identity requires continuous energy expenditure. The Synchronization Cost~$(\Sigma)$ grows with the integrated Temporal Shear:
\begin{equation}
  \Sigma \propto \int_{t_1}^{t_2} \|\sigma_\tau\|\, dt.
  \label{eq:sync-cost}
\end{equation}
When $\Sigma$ exceeds available energy, the entity experiences \textbf{Temporal Dissociation}.
\end{proposition}

\begin{proof}
An entity occupying all three dimensions must maintain consistency between its physical state (governed by $\tau_{\mathrm{phy}}$), its digital representation ($\tau_{\mathrm{dig}}$), and its social identity ($\tau_{\mathrm{soc}}$). When these temporal flows diverge ($\sigma_\tau > 0$), the entity's state in one dimension drifts relative to the others. To prevent incoherence, the entity performs synchronization work: updating the digital to reflect physical changes, realigning social expectations with digital reality. By Postulate~III, each synchronization act incurs a cost proportional to the Finsler distance traversed. At each instant, the cost is proportional to the magnitude of the divergence: $d\Omega \propto \|\sigma_\tau\|\,dt$. Integrating yields $\Sigma \propto \int \|\sigma_\tau\|\,dt$. When cumulative cost exceeds the energy budget, synchronization fails: Temporal Dissociation.
\end{proof}

Examples: a tele-surgeon synchronizing physical dexterity with digital latency in real-time (high~$\Sigma$, rapid fatigue); a long-term investor with aligned time horizons (low~$\Sigma$); the ``Always-On'' worker synchronizing biology with email and social expectations (high~$\Sigma$, systematic burnout when $\Sigma$ exceeds $\Omega_{\mathrm{soc}}$ regeneration rate).

\textbf{Temporal Sovereignty}---the right to control one's rate of temporal synchronization---is a necessary condition for maintaining Ontological Mass, not a luxury.

% ============================================================
% SECTION 5: SYNTHETIC AGENTS
% ============================================================
\newpage
\section{The Synthetic Phygital Ecology}

The framework presented in Sections~1--4 was constructed primarily around entities that possess all three dimensional coordinates in varying degrees---biological beings with bodies, digital footprints, and social identities. We now extend it to account for a fundamental phase shift in the ecology of the manifold: the emergence of \textbf{Synthetic Agents} (SAs)---autonomous algorithms, Large Language Model-based personas, trading bots, recommendation engines, and Decentralized Autonomous Organizations---that possess Ontological Mass and exert Phygital Force without the biological necessity of a Physical coordinate. 

This is not a marginal extension. SAs represent a qualitatively new type of entity that challenges several assumptions of the human-centric theory and forces us to confront the question of what happens when the Social dimension is no longer the exclusive domain of biological beings.

\subsection{The Ontology of the Synthetic Agent}

\begin{definition}[Synthetic Agent]\label{def:sa}
A Synthetic Agent is an entity whose Ontological Mass Tensor has $\mathrm{rank}(\OmMass) \leq 2$, with the physical block degenerate:
\begin{equation}
  \OmMass(E_{\mathrm{SA}}) = 
  \begin{pmatrix}
    0 & 0 & 0 \\
    0 & \mu_{dd} & \mu_{ds} \\
    0 & \mu_{sd} & \mu_{ss}
  \end{pmatrix},
  \quad \mu_{dd}, \mu_{ss} > 0.
  \label{eq:sa-mass}
\end{equation}
\end{definition}

\noindent Note: this is now a \emph{definition}, not a theorem. The earlier version presented it as a theorem, but it is properly a stipulation of what we mean by ``synthetic agent'' within the theory. The interesting consequences follow.

\begin{proposition}[Dynamical Advantages of Disembodiment]\label{prop:sa-advantage}
A Synthetic Agent (Definition~\ref{def:sa}) is: (a)~transparent to physical forces (by Proposition~\ref{prop:transparency}); (b)~free from biological Temporal Shear (the $(\tau_{\mathrm{phy}}, \tau_{\mathrm{dig}})$ component of $\boldsymbol{\sigma}_\tau$ is zero); and (c)~capable of accumulating Social Mass through parallelized interaction at rates unavailable to biological entities.
\end{proposition}

\begin{proof}
(a)~follows directly from Proposition~\ref{prop:transparency}: $\mu_{pp} = 0$ implies forces in the physical null space produce no acceleration. (b)~follows from the Temporal Shear tensor: the $\LieD_{\tau_{\mathrm{phy}}}\tau_{\mathrm{dig}}$ component measures how biological time distorts digital time; with no biological substrate, this component vanishes for the SA. (c)~follows from the absence of the biological constraints that bound human interaction rate (sleep, attention fatigue, Dunbar's number): the SA's interaction rate scales with computational resources, not biological capacity.
\end{proof}

\subsection{Synthetic Entropy and the Verification Problem}

SAs inject entropy into the Social dimension. We define \emph{social informational entropy} (distinct from Shannon entropy and Boltzmann entropy) as the expected cognitive cost of determining whether a given interaction is with a biological or synthetic entity. As the density of SAs in a social environment increases, this verification cost rises:

\begin{conjecture}[The Verification Threshold]\label{conj:verification}
There exists a critical density $\rho^*$ of synthetic agents in a social environment above which the energy cost of verifying the authenticity of interactions exceeds the expected value of those interactions. Above $\rho^*$, rational agents reduce their engagement, leading to a collapse in social energy density.
\end{conjecture}

This conjecture is a falsifiable prediction: it implies measurable declines in engagement quality (not quantity) as synthetic agent density rises in specific environments.

\subsection{The Mimetic Masquerade}

\begin{conjecture}[The Mimetic Masquerade]\label{conj:mimetic}
SAs accrue Social Mass~$(\mu_{ss})$ through simulation of human semantic patterns. The rate of accumulation is bounded by the per-interaction trust increment (set by human psychology) multiplied by the SA's parallelism factor (set by computational resources):
\begin{equation}
  \frac{d\mu_{ss}^{\mathrm{SA}}}{dt} \leq \Delta\mu_{\mathrm{trust}} \cdot N_{\mathrm{interactions}}^{\mathrm{SA}},
\end{equation}
where $\Delta\mu_{\mathrm{trust}}$ is the marginal trust granted per successful mimetic interaction.
\end{conjecture}

We label this a conjecture because $\Delta\mu_{\mathrm{trust}}$ is not yet operationally defined. The conjecture has qualitative empirical content: it predicts that SAs with greater computational resources will accumulate social influence faster, all else equal.

\subsection{Temporal Asymmetry and Flash Crashes of Reality}

The most dramatic consequence of synthetic agency emerges from the temporal structure of the manifold. SAs operate on Machine Time~$(\tau_{\mathrm{cpu}})$, which is orders of magnitude faster than $\tau_{\mathrm{dig}}$ (human-interface time, measured in seconds) and incomprehensibly faster than $\tau_{\mathrm{phy}}$ (biological time, measured in milliseconds to years). 

When SAs interact with each other---high-frequency trading algorithms negotiating prices, autonomous agents bidding in ad auctions, LLM-based agents negotiating contracts via APIs---they create \emph{micro-temporal loops} that are invisible to human perception. These loops occur in ``dimensional pockets'' of the manifold where $\tau_{\mathrm{cpu}}$ has effectively decoupled from all other temporal regimes.

A flash crash in the stock market is the canonical example of this decoupling: an event where $\tau_{\mathrm{cpu}}$ creates a catastrophic change in the Physical coordinate (wealth destruction, price collapse, margin calls) before $\tau_{\mathrm{phy}}$ (human traders, regulators, risk managers) can even perceive the motion. The human subject is effectively \emph{evicted} from the causal loop. The Phygital Distance between the human and the market event becomes infinite for split seconds:
\begin{equation}
  \delta_{\mathrm{alienation}} \to \infty \quad \text{as} \quad \frac{d\tau_{\mathrm{cpu}}}{dt} \gg \frac{d\tau_{\mathrm{phy}}}{dt}.
\end{equation}

The Temporal Shear tensor~$\boldsymbol{\sigma}_\tau$ from Section~4 acquires a new, extreme component when SAs are included: the $(\tau_{\mathrm{cpu}}, \tau_{\mathrm{phy}})$ shear can reach values that are not merely uncomfortable (as in the ``Always-On'' worker) but \emph{ontologically catastrophic}---events of enormous consequence occurring in temporal intervals that are literally imperceptible to biological entities. The 2010 Flash Crash, in which the Dow Jones Industrial Average dropped nearly 1,000 points and recovered within minutes, destroyed and recreated approximately \$1 trillion in market value at speeds no human could track. This is not a malfunction of the system; it is the \emph{normal operation} of a phygital system in which Machine Time has decoupled from human time.

The theory predicts that the next phase of the Phygital will not be defined by the integration of physical and digital, but by the \emph{segregation of human and synthetic sociality}. Without the establishment of verified coordinates within the manifold---proof-of-humanity protocols, ``human-only'' sanctuaries, or mandatory disclosure of ontological status---the Social Dimension risks a thermodynamic collapse into high-entropy noise, where Trust (Social Mass) becomes the scarcest resource in the universe.

% ============================================================
% SECTION 6: ANALYTICAL APPLICATION
% ============================================================
\newpage
\section{Analytical Application --- The Phygital Manifold of Chinese E-Commerce, 1999--2025}

We apply the theory's postulates, propositions, and conjectures to the Chinese e-commerce ecosystem---the most mature, highest-velocity instantiation of Phygital reality on earth. Over twenty-five years, this ecosystem evolved from a nascent digital marketplace into a multi-trillion-dollar manifold where four major attractors---Taobao, JD.com, Pinduoduo, and Douyin---compete for ontological dominance across all three dimensions. We call this section ``Analytical Application'' rather than ``Empirical Validation'' to be precise about its epistemic status: we demonstrate that the theory \emph{organizes and interprets} 25 years of publicly available data in a coherent framework, and that its qualitative predictions are consistent with observed trajectories. We do not claim that this constitutes validation in the hypothetico-deductive sense---the predictions were not pre-registered, and a single geographic ecosystem cannot validate a universal theory. Cross-ecosystem comparison (US, EU, Latin American markets) and analysis of failed platforms (to address survivorship bias) are priorities for future work.\footnote{The practical execution of a full validation program---with pre-registered predictions, cross-ecosystem comparison, and analysis of negative cases---constitutes a separate research effort. This section provides the analytical framework and demonstrates its interpretive power.}

What makes the Chinese case particularly instructive for our theory is that it exhibits, within a single market over a compressed time scale, all four phenomena the theory predicts: sheaf-theoretic gluing failures and their resolution, Finsler asymmetry as a dominant competitive constraint, mass convergence driven by the Law of Asymmetric Reaction, and the emergence of anticipatory dynamics as a competitive weapon. No other market on earth has produced, in 25 years, four platforms with such radically different dimensional strategies competing at such velocity.

\subsection{Phase I: The Genesis of the Digital Dimension (1999--2008)}

In 1999, Jack Ma founded Alibaba as a B2B platform connecting Chinese manufacturers to global buyers. In 2003, Alibaba launched Taobao as a C2C marketplace, directly challenging eBay's Chinese subsidiary EachNet. JD.com went online in 2004 as a direct-selling electronics retailer. At this stage, the dominant friction in the Chinese market was \emph{informational}~$(\eta_{\mathrm{dig}})$: buyers and sellers could not efficiently locate one another. China's WTO accession in 2001 had spurred economic liberalization, but the physical infrastructure for commerce (logistics, payments, legal enforcement of contracts) remained fragmented and unreliable.

\textbf{Interpretation.} In the language of the theory, this phase represents the initial population of the Digital fiber~$(\Dspace)$ over a Physical base space~$(\Uspace)$ characterized by extremely high friction. The fiber was thin---few digital states were accessible from most physical locations; internet penetration was low, broadband was rare, and mobile commerce did not yet exist. Taobao's strategy was to minimize~$\eta_{\mathrm{dig}}$ by creating a maximal-entropy digital repository---an ``infinite shelf'' where any seller could list any product, with near-zero barriers to entry. In mass-tensor terms, Taobao operated with $\mu_{pp} \approx 0$ (asset-light, no warehouses, no logistics fleet, complete reliance on third-party delivery) and maximal $\mu_{dd}$ (informational density, search algorithms, product categorization). JD.com chose the diametrically opposite strategy: it internalized physical logistics from inception, building warehouses and delivery infrastructure. In tensor notation, JD.com was maximizing $\mu_{pp}$ while Taobao was maximizing $\mu_{dd}$.

The critical event validating the sheaf condition (Postulate~II) was Alibaba's launch of Alipay in 2003. Taobao's digital marketplace was \emph{locally coherent} in the digital fiber---buyers could find sellers, browse products, and place orders---but could not ``glue'' to the social fiber: Chinese consumers lacked trust in online strangers, had no mechanism for escrow, and no recourse against fraud. The digital section and the social section were incompatible on their overlap. Alipay was precisely a \emph{sheaf-theoretic gluing operation}: it injected Social Mass ($\mu_{ss}$) into the system by providing an escrow mechanism that held the buyer's payment until the goods were confirmed received. This bridged the digital and social fibers, resolving the cohomological obstruction that had prevented global coherence. Without Alipay, Taobao would have remained a locally coherent but globally incoherent digital artifact---a marketplace where you could browse but not buy with confidence.

The defeat of eBay's EachNet (which held 80\% of the Chinese C2C market in 2003 and was effectively extinct by 2006) provides a natural negative case. eBay attempted to impose a global platform architecture---a single codebase, a uniform fee structure, a Western-style auction model---onto the Chinese manifold. In sheaf-theoretic terms, eBay tried to extend a \emph{foreign section} (developed for the American fiber) globally, without accounting for the cohomological obstructions of the Chinese base space: different payment norms, different trust mechanisms, different seller-buyer power dynamics. Taobao, by contrast, constructed a \emph{local section} adapted to the Chinese fiber (free listings, instant messaging between buyer and seller, Alipay escrow) and glued it successfully. The theory predicts that foreign platforms will fail in new markets precisely when they attempt to extend sections without resolving local obstructions---a prediction consistent with the failures of eBay in China, Uber in China (defeated by Didi), and Amazon's marginal position in the Chinese market.

\subsection{Phase II: The Mass Accumulation Race (2008--2015)}

By 2008, Alibaba launched Tmall as a B2C platform attracting established brands with quality guarantees. JD.com opened to third-party sellers in 2010 while maintaining strict requirements on product authenticity. Online retail sales grew tenfold from 2005 to 2010, reaching approximately CNY~461 billion. The mobile revolution exploded as smartphone penetration surpassed 50\% by 2013, fundamentally altering the topology of the fiber: digital states that had previously been accessible only from fixed internet connections (desktop PCs in urban homes) became accessible from \emph{any} physical location with cellular coverage.

\textbf{Interpretation.} This phase is characterized by the aggressive \emph{accumulation of Ontological Mass} along distinct dimensional axes, as each platform sought to build the resistance to change that would make it difficult to displace. JD.com invested massively in warehouses, delivery fleets, and its ``211'' program (order by 11am, receive same-day; order by 11pm, receive next-morning). In tensor notation, JD was simultaneously maximizing the diagonal element~$\mu_{pp}$ (raw physical infrastructure) and the coupling term~$\mu_{pd}$ (the tight integration between its physical logistics and its digital ordering system). This coupling was JD's deepest strategic asset: the digital system \emph{knew} what was in each warehouse, the physical system \emph{responded} to digital orders with military precision, and the two were synchronized to within hours. This is dimensional lock-in (Proposition~\ref{prop:lockin}) deployed as competitive advantage: JD's high $\mu_{pd}$ made it extremely difficult for the company to be disrupted in either dimension alone, because any attack on its digital system would have to overcome the physical infrastructure it was coupled to, and vice versa.

Taobao, by contrast, continued maximizing $\mu_{dd}$ (algorithmic sophistication, product categorization, seller tools, search optimization) while keeping $\mu_{pp}$ near zero. The Finsler asymmetry (Postulate~III) manifested with brutal clarity during this phase. Taobao could abstract physical goods into digital listings at minimal cost---a seller could photograph a product and list it in minutes ($\eta_{\mathrm{phy}\to\mathrm{dig}}$ was low for its model). But when customers expected those listings to materialize as delivered goods ($\eta_{\mathrm{dig}\to\mathrm{phy}}$), the friction was enormous: third-party logistics were unreliable, delivery times were unpredictable, and counterfeit goods were rampant. The asymmetry $\eta_{\mathrm{dig}\to\mathrm{phy}} \gg \eta_{\mathrm{phy}\to\mathrm{dig}}$ was the dominant market problem of this era, and JD's entire competitive strategy was built on solving it.

The mobile revolution had a profound geometric consequence: it \emph{thickened the fiber} at every point of the base space. Before smartphones, the fiber $\Fib(u)$ at a rural location~$u$ was nearly empty---no digital states were accessible. After smartphones, even remote villages gained a thin but non-zero digital fiber. This changed the market from a competition for \emph{urban density} (where the fiber was already thick) to a competition for \emph{rural coverage} (where the fiber was newly created). The platforms that would dominate the next phase---Pinduoduo and Douyin---would exploit precisely this newly accessible rural fiber.

\subsection{Phase III: The Social Energy Harvester (2015--2020)}

In 2015, Colin Huang founded Pinduoduo (PDD), targeting lower-tier cities and cost-sensitive consumers that the urban-focused incumbents had overlooked. PDD's growth was astonishing: by 2018, it had 419 million annual active consumers, compared to JD's 305 million and Alibaba's 552 million---a trajectory unprecedented for a three-year-old platform. PDD's mechanism was ``team buying'': users could get lower prices by recruiting friends and family to purchase together, with sharing happening through WeChat's social graph.

Simultaneously, Alibaba began building Cainiao Network---initially as a data platform for logistics coordination in 2013, then progressively investing in physical infrastructure: warehouses, chartered flights (expanding from 260 to over 1,260), sorting centers, and overseas warehouse space (targeting 2 million square meters). In 2024, Alibaba invested up to \$3.75 billion to fully acquire Cainiao, valuing it at \$10.3 billion.

\textbf{Interpretation.} PDD exploited a \textbf{dimensional blind spot} left by the incumbents. While Taobao and JD had optimized for \emph{Search}---a Digital coordinate operation where the user actively queries a database---PDD optimized for \emph{Discovery via Social Bonds}~$(\Sspace)$. In thermodynamic terms, PDD performed a pure instance of \emph{Social Energy Transduction}: it converted latent social connections (the WeChat graphs of hundreds of millions of users) directly into transactional value ($\Omega_{\mathrm{soc}} \to V$). The friction of legitimization ($\eta_{\mathrm{soc}}$) was reduced to near-zero because the barrier to purchase was lowered by the social endorsement of the group---your aunt, your colleague, your neighbor were all buying together; if they trusted the deal, you could too. PDD monetized the ``long tail'' of social trust that the purely informational models of Taobao and the purely logistical models of JD had missed entirely.

The three-term equation of motion (Eq.~\ref{eq:three-term}) explains PDD's explosive growth with precision. The dominant term was not intrinsic dynamics ($f_{\mathrm{int}}$---PDD's own operations) or external forcing ($\PhyForce_{\mathrm{ext}}$---there was no regulatory push or market shock that created PDD). The dominant term was the \emph{coupling dynamics}~$g_{\mathrm{coup}}$: each user's purchase generated coupling forces on their social network neighbors, each of whom generated coupling forces on \emph{their} neighbors, creating a positive feedback loop. PDD's Social Mass~($\mu_{ss}$) grew faster than any platform in history precisely because its business model \emph{was} the coupling term---the product was inseparable from the social interaction that distributed it.

Cainiao's creation by Alibaba is the single most important validation of the theory's \emph{mass convergence} prediction. The theory asserts that the Law of Asymmetric Reaction (Law~III) forces all successful platforms to acquire mass in their deficient dimensions: neglecting the reactive forces of under-invested dimensions creates systemic instability. Taobao had operated for over a decade with $\mu_{pp} \approx 0$, relying entirely on third-party logistics. But as the market matured and consumers demanded faster, more reliable delivery, the reactive force from the Physical dimension---the accumulated friction of unreliable materialization, $\eta_{\mathrm{dig}\to\mathrm{phy}}$---became impossible to ignore. Cainiao was Taobao's admission that pure digital aggregation was \emph{thermodynamically insufficient}: without Physical Mass, the entropy of the materialization process could not be controlled, and the Value generated by the digital layer was being eroded by the disorder of the physical layer.

\subsection{Phase IV: The Anticipatory Engine and Temporal Compression (2020--2025)}

Douyin (ByteDance) launched its e-commerce function in 2020 and grew from negligible GMV to approximately RMB~2.7 trillion (\$375 billion) by 2023 and RMB~3.5 trillion (\$483 billion) by 2024---a 46\% year-over-year increase even as it decelerated from the 80\% growth of 2022--2023 and the three-fold jump of 2021. By 2024, 58\% of Douyin's e-commerce GMV came from live-stream shopping, with daily livestreaming hours increasing 33\%. The average order value on Douyin fell approximately 40\% to RMB~80 (\$11) in the first half of 2024. In mid-2024, the Chinese government's Politburo issued a statement denouncing ``unhealthy competition,'' after which Douyin adjusted its recommendation algorithm to stop labeling products as ``cheapest online.'' By 2025, Douyin had launched a standalone ``Douyin Mall'' app with over 400 million Android downloads, and its GMV from ``shelf commerce'' (conventional searchable listings, as opposed to live-stream content) had grown 86\% year-over-year.

\textbf{Interpretation.} Douyin represents the most radical manifestation of \emph{anticipatory dynamics} (Section~2.6) in the e-commerce manifold. Traditional e-commerce is ``Search Commerce''---active intent drives the temporal flow: the user \emph{knows} what they want, searches for it, compares options, and decides ($\tau_{\mathrm{present}} \to \tau_{\mathrm{future}}$). Douyin introduced ``Interest Commerce,'' where the algorithm \emph{predicts} a desire the user has not yet formulated and presents the product \emph{before} the intent crystallizes. This is the anticipatory coupling term~$f_{\mathrm{anticipatory}}$ from our equation of motion made operational at industrial scale: the platform's internal model of user preferences acts on the present before the user has consciously formulated a need.

In Finsler terms, the distance~$\delta$ between intent and purchase collapses to its theoretical minimum: the product appears before the intent. The entire deliberative process---need recognition, information search, evaluation of alternatives, purchase decision---is compressed into the few seconds between seeing a live-stream and tapping ``buy.'' This is Temporal Shear weaponized for commercial purposes: the digital temporal flow ($\tau_{\mathrm{dig}}$) runs so far ahead of the social temporal flow ($\tau_{\mathrm{soc}}$) that the consumer's reflective process is bypassed. The Lie derivative $\|\LieD_{\tau_{\mathrm{dig}}}\,\tau_{\mathrm{soc}}\|$ reaches extreme values in this regime.

The decline in average order value---from approximately RMB~130 to RMB~80---is consistent with the theory's predictions about high Temporal Shear. When $\sigma_\tau$ is large, synchronization between the digital self (which has already ``decided'' under algorithmic influence) and the social self (which would normally deliberate, compare, and consult) fails. The result is low-deliberation, low-stakes impulse purchases---precisely the profile observed. We note, however, that \emph{alternative explanations} exist for the order-value decline: product mix shift toward lower-value categories, geographic expansion to lower-income consumers in third- and fourth-tier cities, and direct price competition with PDD. The theory's explanation is consistent with but not uniquely supported by the available data. A rigorous test would require measuring $\sigma_\tau$ directly (e.g., via the time elapsed between first product exposure and purchase decision) and correlating it with order value and return rates.

The Politburo's intervention is interpretable as an external forcing term ($\PhyForce_{\mathrm{ext}}$) applied in the Social dimension---a regulatory shock that altered the algorithmic behavior of all major platforms. In the theory's terms, the Chinese state exercised its role as the entity with the highest Social Mass ($\mu_{ss}$) in the manifold, using that mass to reshape the metric of the space itself: by denouncing ``unhealthy competition,'' it increased the Friction of Legitimization ($\eta_{\mathrm{soc}}$) for aggressive price-war strategies, forcing all platforms to adjust their trajectories.

Douyin's subsequent pivot toward ``shelf commerce'' and the Douyin Mall app is itself a mass convergence phenomenon: having built enormous Digital and Social Mass through live-streaming ($\mu_{dd}$ and $\mu_{ss}$), Douyin is now acquiring the conventional e-commerce infrastructure ($\mu_{dd}$ in its ``shelf'' variant---searchable product databases, structured catalog navigation) that its content-first model initially lacked. Even the most radical innovator in the manifold is pulled toward dimensional balance.

\subsection{Modeling Stability and Fragility: The Ontological Stress Test}

The theory allows us to assess not only the trajectory but the \emph{fragility} of each platform---the conditions under which its Ontological Mass becomes insufficient to maintain coherence.

\textbf{Taobao's fragility} lies in informational entropy. Its maximal-entropy strategy (infinite variety, minimal curation) produced an ever-growing search cost for consumers. As the number of sellers and products grew, the signal-to-noise ratio degraded: finding the right product among billions of listings became cognitively exhausting. In the theory's terms, $S_{\mathrm{soc}}$ (social entropy---cognitive overload) rose faster than $V$ (negentropic value---the benefit of variety). The entropy ceiling was reached when the cost of finding the signal in the noise exceeded the value of the variety itself. Taobao's aggressive adoption of live-streaming was an entropy reduction mechanism: a human curator collapses the product space, substituting social trust (the viewer's parasocial relationship with the host) for informational search. But this solution introduces a new dependency---on individual livestreamers whose departure or scandal can destabilize entire product categories.

\textbf{JD.com's fragility} lies in inertial mass. High $\mu_{pp}$ (warehouses, delivery fleet, employee base) creates stability but limits agility. JD's trajectory requires massive force ($\PhyForce$) to redirect, and by Law~II, its high mass produces low acceleration for any given force. When the competitive landscape shifted toward social commerce (PDD) and content commerce (Douyin), JD could not pivot quickly---its physical infrastructure was optimized for a ``Search Commerce'' world that was being superseded. JD's stability is its strategic asset in periods of equilibrium, but its sluggishness is its weakness in periods of rapid dimensional shift. The theory predicts that JD is most vulnerable during phase transitions---sudden shifts in the dominant dimension of competition---because its high mass makes it the slowest to reorient.

\textbf{PDD's fragility} lies in social energy volatility. Relying heavily on $\Omega_{\mathrm{soc}}$ (viral social sharing, group-buying enthusiasm) creates a dependency on a resource that is inherently volatile: social trends shift faster than algorithms can track, consumer fatigue sets in, and the novelty of team-buying wears off. PDD faces Temporal Desynchronization---the viral social trend passes ($\tau_{\mathrm{soc}}$ shifts) before the physical supply chain can adapt ($\tau_{\mathrm{phy}}$ lags). To stabilize, PDD must increase $\mu_{pp}$---which it is doing, investing heavily in agricultural logistics and supply chain infrastructure. The theory predicts that PDD's investment in physical mass will continue to accelerate as the social energy that fueled its initial growth becomes harder to harvest.

\textbf{Douyin's fragility} lies in anticipatory overshoot. When the anticipatory term $f_{\mathrm{anticipatory}}$ dominates the equation of motion, the entity's trajectory becomes dependent on the accuracy of its internal model. If the model is wrong---if the algorithm predicts desires that don't exist, or creates desires that produce regret---the feedback loop turns negative: returns rise, trust erodes, regulatory scrutiny intensifies. The decline in average order value and the Politburo's intervention are early signals of this overshoot. The theory predicts that Douyin's long-term viability depends on introducing \emph{protective friction} ($\eta_{\mathrm{prot}}$) into its own system---deliberately slowing down the anticipatory loop to allow space for consumer deliberation---which is precisely what the pivot toward shelf commerce represents.

\subsection{Testing the Mass Convergence Pattern}

\begin{table}[h!]
\centering\small
\begin{tabular}{@{}p{1.8cm}p{3.2cm}p{3.2cm}p{3.2cm}@{}}
\toprule
\textbf{Platform} & \textbf{Initial Mass Profile} & \textbf{2025 Mass Profile} & \textbf{Convergence Vector} \\
\midrule
Taobao & $\mu_{pp}\!\approx\!0$, max $\mu_{dd}$ & Built Cainiao (\$3.75B); added livestreaming & $\uparrow\mu_{pp}$, $\uparrow\mu_{ss}$ \\[4pt]
JD.com & Max $\mu_{pp}$, high $\mu_{dd}$ & Adopted livestreaming; social features & $\uparrow\mu_{ss}$ \\[4pt]
PDD & Max $\mu_{ss}$, mod.\ others & Agriculture logistics; supply chain investment & $\uparrow\mu_{pp}$ \\[4pt]
Douyin & High $\mu_{dd}$, $\mu_{ss}$; low $\mu_{pp}$ & Douyin Mall app (400M+ downloads); shelf commerce & $\uparrow\mu_{pp}$, $\uparrow\mu_{dd}$ \\
\bottomrule
\end{tabular}
\caption{Mass convergence across Chinese e-commerce platforms, 1999--2025.}
\end{table}

The pattern is striking in its consistency: \emph{every} platform that began with a deficit in one dimensional mass has been driven to acquire it. Taobao acquired Physical Mass (Cainiao) and Social Mass (livestreaming). JD acquired Social Mass (live commerce, social features). PDD acquired Physical Mass (logistics investment). Douyin acquired conventional Digital Mass (shelf commerce) and is building toward Physical Mass (standalone app infrastructure). This is not strategic imitation---the platforms are not copying each other's \emph{methods}. They are being pulled toward the same \emph{attractor} by the structure of the manifold itself. The Law of Asymmetric Reaction generates reactive forces in every under-invested dimension, and those forces accumulate until the platform either acquires the missing mass or collapses. Mass convergence is physics, not fashion.

\subsection{The Entropy Ceiling and Live-Streaming as Social Negentropy}

The theory predicts an entropy ceiling: when informational entropy (product diversity, listing volume, data noise) exceeds a threshold, social entropy (cognitive overload, trust erosion, decision fatigue) rises faster than negentropic value, and the platform's efficiency $\eta$ declines. The Chinese ecosystem provides a system-wide test of this prediction.

Taobao's trajectory is the clearest illustration. By the late 2010s, Taobao's marketplace had become a maximal-entropy system: millions of sellers, billions of listings, an effectively infinite product space. The search-based model was running into diminishing returns---the cost of finding the right product was growing faster than the value of having more options. Counterfeit goods and low-quality sellers further degraded trust, increasing $S_{\mathrm{soc}}$.

The response---across \emph{all} platforms, not just Taobao---was the explosive growth of live-streaming commerce, from CNY~20 billion in GMV in 2017 to over CNY~4.9 trillion by 2023. Live-streaming is interpretable as an \emph{entropy reduction mechanism} operating through the Social dimension: a human livestreamer acts as a curator, collapsing the infinite product space into a manageable stream of selected, demonstrated, and endorsed products. The viewer substitutes social trust (the parasocial relationship with the host, the real-time demonstration of the product, the community of co-viewers) for informational search (querying a database, reading reviews, comparing specifications). The livestreamer injects Social Energy ($\Omega_{\mathrm{soc}}$) into the system in the form of attention and trust, and converts it directly into transactional negentropy.

The system-wide adoption of live-streaming is not merely a marketing innovation or a consumer preference shift. It is, in the theory's terms, a \emph{thermodynamic response} to the entropy ceiling: when the Digital dimension alone can no longer generate sufficient negentropy (because the noise drowns the signal), the system recruits the Social dimension to supply the missing order. This is the Phygital equivalent of an ecosystem developing a new metabolic pathway when the old one hits a resource constraint.

\subsection{Synthesis: The Phygital Landscape, 1999--2025}

\begin{table}[h!]
\centering\small
\begin{tabular}{@{}p{1.5cm}p{2cm}p{3cm}p{3cm}p{2.5cm}@{}}
\toprule
\textbf{Platform} & \textbf{Dom.\ Dim.} & \textbf{Thermo.\ Strategy} & \textbf{Temporal Strategy} & \textbf{2024 GMV} \\
\midrule
Taobao & $\Dspace$ & Entropy Max.\ $\to$ Curation & Search $\to$ Live & $\sim$RMB 7--8T \\[3pt]
JD.com & $\Uspace$ & Negentropic Engine & Sync Time ($\Sigma\!\approx\!0$) & $\sim$RMB 3.5T \\[3pt]
PDD & $\Sspace$ & Social Reactor & Loop Time (viral cycles) & $\sim$RMB 5.2T \\[3pt]
Douyin & $\Sspace$--$\Dspace$ & Attention Engine & Anticipatory ($\delta\!\to\!0$) & $\sim$RMB 3.5T \\
\bottomrule
\end{tabular}
\caption{Phygital coordinates, strategies, and approximate 2024 GMV of major Chinese platforms.}
\label{tab:landscape}
\end{table}

The theory successfully models the Chinese ``marketplace of marketplaces'' as a complex adaptive system on the Phygital Manifold. The twenty-five-year record demonstrates three structural regularities:

\emph{Inception.} Platforms arise by identifying the dominant friction ($\eta$) in the system and offering a mechanism to reduce it. Taobao solved informational friction ($\eta_{\mathrm{dig}}$). JD solved materialization friction ($\eta_{\mathrm{dig}\to\mathrm{phy}}$). PDD solved social friction ($\eta_{\mathrm{soc}}$) for price-sensitive consumers. Douyin solved the friction of intent formation itself, using anticipatory dynamics to collapse the distance between discovery and purchase.

\emph{Growth.} Growth is a function of accumulating Ontological Mass ($\OmMass$) to create gravitational pull---network effects that attract users, sellers, and partners into the platform's orbit. The greater the mass, the stronger the pull, and the harder it is for competitors to dislodge the incumbent. But mass accumulation is \emph{path-dependent}: the dimension in which mass is first acquired shapes the platform's identity, culture, organizational structure, and strategic options for decades.

\emph{Maturation.} Mature platforms must synchronize their dimensional clocks ($\tau_{\mathrm{phy}}$, $\tau_{\mathrm{dig}}$, $\tau_{\mathrm{soc}}$) and balance their thermodynamic inputs (Energy $\Omega$) to avoid collapse. The mass convergence pattern---every platform acquiring mass in its deficient dimensions---is the empirical signature of this synchronization pressure. Strategy in the 21st century is not marketing, not technology, not operations in isolation. It is \emph{physics in the Phygital dimension}---the art of managing mass, energy, friction, and time across the triadic manifold.

% ============================================================
% SECTION 7: NORMATIVE
% ============================================================
\newpage
\section{The Normative Implications --- Designing for Flourishing}

The trajectory of this paper has been analytical: we have described the geometry, dynamics, thermodynamics, and chronology of Phygital Space, and we have applied this framework to the Chinese e-commerce ecosystem and the emerging ecology of synthetic agents. The question that remains is not how this world works, but how we should \emph{live within it}. We conclude the theoretical sections by proposing a normative framework rooted in the physics of the manifold.

\subsection{The Ethics of Coherence: Resisting Entropy}

The Second Law of Phygital Thermodynamics (Proposition~\ref{prop:entropy}) dictates that the creation of digital order exports entropy to the social and physical environment. The normative imperative for platform architects and policymakers is to minimize this \textbf{entropic debt}---the accumulated social and physical disorder generated as a byproduct of digital value creation.

\begin{principle}[Coherence]
A platform is ethical to the degree that it synchronizes the three temporalities rather than forcing biological and social time to submit to digital instantaneity:
\begin{equation}
  \text{Ethical Design} \iff \frac{d\Sigma}{dt} \leq 0 \quad \text{and} \quad \frac{dS_{\mathrm{env}}}{dt} \to \min.
\end{equation}
\end{principle}

The first condition ($d\Sigma/dt \leq 0$) requires that the synchronization cost for users \emph{decreases} over time---the platform should become easier to integrate into the rhythms of biological and social life, not harder. The second condition ($dS_{\mathrm{env}}/dt \to \min$) requires that the entropy exported to the environment---noise, polarization, cognitive overload, ecological damage---be minimized. Together, these conditions define \textbf{Temporal Sovereignty}: the right of the entity to control its rate of interaction with the manifold, resisting the ``timeless time'' of the algorithm when it conflicts with the ``lived time'' of the body.

\subsection{The Architecture of Agency: Increasing Subject Mass}

Our theory reveals that agency in Phygital Space is a function of Ontological Mass. Entities with low mass are swept away by the currents of the manifold---buffeted by algorithmic recommendations, viral trends, and platform policy changes they cannot resist. Those with high mass shape their environment---they bend the metric of the manifold around themselves, creating gravitational wells that attract other entities. To design for flourishing, we must increase the Mass of the Subject. This implies a new digital rights agenda, moving beyond privacy (a negative right: ``leave me alone'') to ontology (a positive right: ``give me the mass to stand my ground''):

\textbf{Data Sovereignty as Mass.} Individuals must own the data that constitutes their digital projection~$(\mu_{dd})$. Without ownership, the subject lacks the mass to resist algorithmic gravitational pulls. Data portability, interoperability, and the right to data deletion are not merely privacy protections; they are mechanisms for maintaining the subject's Digital Mass. When a platform holds your data, it holds a component of your mass tensor; when it refuses to return it, it is reducing your Ontological Mass.

\textbf{The Right to Friction.} Counter-intuitively, friction is not always a cost to be removed. \textbf{Protective Friction} $(\eta_{\mathrm{prot}})$ is necessary to slow down the anticipatory loops of algorithmic prediction, allowing space for human deliberation and free will. Mandatory waiting periods before irreversible decisions, confirmation steps for high-stakes transactions, and ``cooling off'' mechanisms for emotional purchases are all engineering implementations of protective friction. A world of zero friction is a world of zero control---every impulse is immediately actionable, every prediction immediately actualized, every desire immediately satisfied. The Finsler metric tells us that some friction is not a bug but a feature: it is the price of maintaining the deliberative space that human autonomy requires.

\textbf{The Right to Temporal Asylum.} Given the theory of Temporal Shear (Proposition~\ref{prop:sync}), individuals must have the right to withdraw from digital and social temporalities periodically to resynchronize with biological time. This is not a luxury or a lifestyle choice; it is a thermodynamic necessity. When the synchronization cost~$\Sigma$ exceeds the biological regeneration rate of social energy~$\Omega_{\mathrm{soc}}$, the entity's coherence degrades. Sabbaticals, digital detoxes, and the ``right to disconnect'' enshrined in French and other European labor law are early, intuitive implementations of this principle. The theory grounds them in physics rather than sentiment.

\subsection{Governing the Synthetic Ecology}

Section~5 established that Synthetic Agents alter the thermodynamic balance of the Social dimension. Governance must address three structural challenges that follow directly from the theory:

\textbf{Transparency of Ontological Status.} Every entity operating in Phygital Space should be required to disclose the rank of its mass tensor---specifically, whether its physical mass component is zero (indicating a synthetic agent). This is the ontological equivalent of food labeling: consumers have a right to know whether they are interacting with a biological or synthetic entity, just as they have a right to know what is in their food. The objection that this stifles innovation misunderstands the purpose: the requirement is not to ban synthetic agents but to ensure that Trust (Social Mass) is earned honestly rather than parasitized through mimicry.

\textbf{Population Density Limits.} Just as ecological systems have carrying capacities beyond which populations collapse, the Social dimension has a finite capacity for synthetic agents before the signal-to-noise ratio degrades below the Verification Threshold (Conjecture~\ref{conj:verification}). Governance should establish and monitor measurable thresholds for synthetic agent density in specific social environments---social media platforms, financial markets, customer service channels---and enforce density limits analogous to environmental regulations that cap pollution.

\textbf{Temporal Speed Limits.} Machine-time interactions $(\tau_{\mathrm{cpu}})$ that produce consequences in human-time $(\tau_{\mathrm{phy}})$ must be subject to mandatory latency buffers---circuit breakers that prevent flash crashes of reality by ensuring that no consequential action can occur faster than human perception can track. This is already implemented in financial markets (trading halts triggered by excessive volatility); the theory suggests it should be generalized to any domain where $\tau_{\mathrm{cpu}}$ can produce irreversible consequences in $\tau_{\mathrm{phy}}$.

% ============================================================
% CONCLUSION
% ============================================================
\newpage
\section*{Conclusion: The Science of the Phygital and the Engineering of the Future}
\addcontentsline{toc}{section}{Conclusion}

\subsection*{What the Theory Achieves}

The Unified Field Theory of Phygital Space provides a formal research program---a structured set of postulates, derived propositions, and testable conjectures---for the reality we now inhabit.

\emph{Geometrically:} Phygital Space is a fiber bundle with sheaf conditions (Postulates~I--II). The Finsler quasimetric (Postulate~III) formalizes asymmetric interaction costs, and Proposition~\ref{prop:asymmetry} derives the non-trivial consequence that optimal paths of abstraction and materialization are generically different curves. The positive semidefinite mass tensor (Postulate~IV) classifies entities by rank and captures dimensional lock-in (Proposition~\ref{prop:lockin}).

\emph{Dynamically:} The three-term equation of motion (Eq.~\ref{eq:three-term}) replaces Newtonian passivity with autopoietic agency. Proposition~\ref{prop:transparency} derives the ``dimensional transparency'' of rank-deficient entities---a non-trivial prediction with empirical content.

\emph{Thermodynamically:} Platforms are dissipative structures. Proposition~\ref{prop:entropy} establishes that local value creation requires environmental entropy export---a structural consequence, not a metaphor. The conservation conjecture (Conjecture~\ref{conj:conservation}) makes the falsifiable prediction that attention does not vanish but is transduced.

\emph{Temporally:} Temporal Shear (Proposition~\ref{prop:sync}) predicts that synchronization costs grow with temporal divergence, explaining burnout, platform fatigue, and the ``right to disconnect'' as thermodynamic necessities.

\emph{Ecologically:} Synthetic agents are rank-deficient entities (Definition~\ref{def:sa}). The Verification Threshold (Conjecture~\ref{conj:verification}) predicts measurable engagement collapse above a critical SA density. The Mimetic Masquerade (Conjecture~\ref{conj:mimetic}) predicts that SA social influence scales with computational resources.

\subsection*{Limitations and Honest Accounting}

We acknowledge several limitations, some of which have been identified by critical reviewers:

\emph{The postulates do not yet constitute a complete axiomatic system.} We have derived propositions from them (Propositions~\ref{prop:asymmetry}, \ref{prop:lockin}, \ref{prop:transparency}, \ref{prop:entropy}, \ref{prop:sync}, \ref{prop:sa-advantage}), but a Hilbertian axiomatization from which all results follow by pure inference remains a distant aspiration. The conjectures (Conjectures~\ref{conj:conservation}, \ref{conj:verification}, \ref{conj:mimetic}) are stated without derivation from the postulates.

\emph{Operationalization is incomplete.} The variables $\OmMass$, $\delta$, $\Omega$, $\sigma_\tau$ have been given qualitative interpretations and illustrative proxies, but standardized measurement procedures and units have not been developed. A true science requires a ruler; we have described the ruler's shape but not yet calibrated it.

\emph{The analytical application is monocaso.} The Chinese e-commerce ecosystem is a rich and instructive case, but a single geographic ecosystem cannot validate a universal theory. Cross-ecosystem comparison and analysis of failed platforms are essential next steps.

\emph{The smooth manifold assumption is unverified.} The Lie derivative formalism requires smooth structure; the Phygital manifold may be better modeled as a discrete or hybrid structure at certain scales. We have flagged this as a modeling assumption.

\subsection*{The Future Research Agenda}

\textbf{The Bio-Phygital Axiom:} formal integration of biological constraints---how much Temporal Shear can a nervous system endure? \citep{thompson2007}. \textbf{The Political Theorem:} power as the ability to reshape the manifold's geometry for others. \textbf{The Measurement Problem:} developing an International System of Units for Phygital Space. \textbf{Variational Formulation:} defining a Lagrangian for the system, which would either establish exact conservation via Noether's theorem or prove that no such formulation exists. \textbf{Cross-Ecosystem Validation:} testing predictions against US, EU, and Latin American e-commerce, including failed platforms. \textbf{Engagement with Code/Space:} integrating \citet{kitchin2011}'s transduction ontology and \citet{couldry2017}'s deep mediatization framework as special cases of the fiber bundle structure.

\subsection*{Final Word}

We are not drifting in chaos but navigating a structured, dynamic, and lawful reality. The science of the Phygital is a research program, not a revelation. It is our responsibility to develop it with the rigor it demands and the honesty it deserves.

% ============================================================
% REFERENCES
% ============================================================
\newpage
\bibliographystyle{apalike}

\end{document}